\documentclass[cleveref,a4paper,english,nolineno]{socg-lipics-v2021}
\usepackage{graphicx} %
\usepackage{xcolor}
\usepackage{enumitem}
\usepackage{inconsolata}

\usepackage{amsmath,amssymb,amsthm}
\hypersetup{
    colorlinks=true,
    allcolors=blue,
}

\usepackage{stmaryrd}
\newcommand{\mapsot}{\mapsfrom}

\Crefname{claim}{Claim}{Claims}     %

\usepackage{algorithm}
\usepackage{algorithmicx}
\usepackage[noend]{algpseudocode}

\usepackage{booktabs}
\newcommand{\ra}[1]{\renewcommand{\arraystretch}{#1}}

\usepackage[normalem]{ulem} %

\newcommand{\Rsp}{\mathbb{R}}
\newcommand{\Xsp}{\mathbb{X}}
\newcommand{\Zsp}{\mathbb{Z}}

\newcommand{\ot}{\leftarrow}

\renewcommand{\paragraph}[1]{\subparagraph*{#1}}

\newcommand{\im}{\operatorname{im}}

\newcommand{\Hgr}{{\sf H}}

\newcommand{\prism}[1]{\overline{#1}}
\newcommand{\cone}[1]{\hat{#1}}

\newcommand{\bdry}{\partial}
\newcommand{\cbdry}{\delta}
\newcommand{\low}{\operatorname{low}}

\newcommand{\pK}{\prism{K}}
\newcommand{\cK}{\cone{K}}
\newcommand{\pz}{\prism{z}}
\newcommand{\pf}{\prism{f}}

\newcommand{\subK}[1]{K_{-1}^{#1}}
\newcommand{\supK}[1]{K_{#1}^{n+1}}

\newcommand{\ssx}{\sigma}
\newcommand{\tsx}{\tau}
\newcommand{\cssx}{\cone{\ssx}}
\newcommand{\ctsx}{\cone{\tsx}}

\newcommand{\coeff}[2]{\langle #1, #2 \rangle}

\newcommand{\algname}[1]{\textsc{#1}}
\newcommand{\winit}{w_\textrm{init}}
\newcommand{\wfinal}{w_\textrm{final}}

\newcommand{\ee}{\varepsilon}

\newcommand{\cancel}[1]{}

\title{Apex Representatives}
\author{Tamal K. Dey}
{Purdue University}
{tamaldey@purdue.edu}{https://orcid.org/0000-0001-5160-9738}{}
\author{Tao Hou}
{University of Oregon}
{taohou@uoregon.edu}{https://orcid.org/0000-0002-3389-6136}{}
\author{Dmitriy Morozov}
{Lawrence Berkeley National Laboratory}
{dmitriy@mrzv.org}{https://orcid.org/0000-0002-4330-6670}{}

\authorrunning{T.\ Dey, T.\ Hou, D.\ Morozov}
\Copyright{Tamal Dey, Tao Hou, Dmitriy Morozov}

\ccsdesc[500]{Theory of computation~Computational geometry}
\ccsdesc[500]{Mathematics of computing~Algebraic topology}
\keywords{zigzag persistent homology, Mayer--Vietoris pyramid, cycle representatives}

\funding{T.\ K.\ Dey was partially supported by NSF grants CCF-2437030 and DMS-2301360. 
T.\ Hou was supported by NSF fund CCF 2439255.
         D. Morozov was supported by the Director, Office of Science, Office of
         Advanced Scientific Computing Research, of the U.S.\ Department of
         Energy under Contract No.\ DE-AC02-05CH11231.}

\EventEditors{Oswin Aichholzer and Haitao Wang}
\EventNoEds{2}
\EventLongTitle{41st International Symposium on Computational Geometry
(SoCG 2025)}
\EventShortTitle{SoCG 2025}
\EventAcronym{SoCG}
\EventYear{2025}
\EventDate{June 23--27, 2025}
\EventLocation{Kanazawa, Japan}
\EventLogo{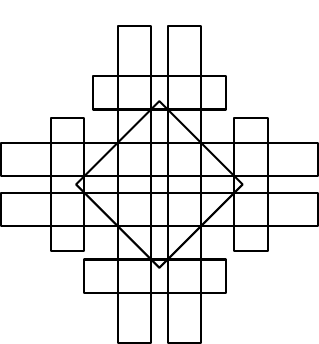}
\SeriesVolume{332}
\ArticleNo{XX}     %

\begin{document}

\maketitle

\begin{abstract}
    Given a zigzag filtration, we want to find its barcode representatives,
    i.e., a compatible choice of bases for the homology groups that diagonalize
    the linear maps in the zigzag.  To achieve this, we convert the input zigzag
    to a levelset zigzag of a real-valued function. This function generates a
    Mayer--Vietoris pyramid of spaces, which generates an infinite strip of
    homology groups. We call the origins of indecomposable (diamond) summands of
    this strip their apexes and give an algorithm to find representative cycles
    in these apexes from ordinary persistence computation. The resulting
    representatives map back to the levelset zigzag and thus yield barcode
    representatives for the input zigzag. Our algorithm for lifting a
    $p$-dimensional cycle from ordinary persistence to an apex representative
    takes $O(p \cdot m \log m)$ time.
    From this we can recover zigzag representatives in time $O(\log m + C)$,
    where $C$ is the size of the output.
\end{abstract}

\vspace{-1ex}
\begin{figure}[htbp]
    \centering
    \includegraphics{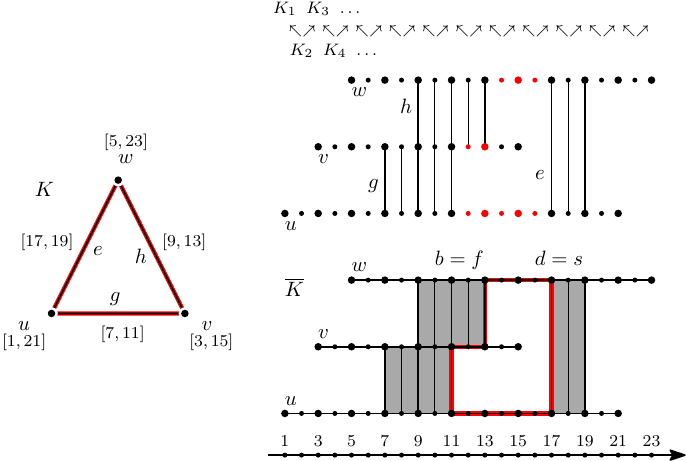}
    \caption{A zigzag on the top. A total complex $K$ on the left. The support of each simplex in a zigzag,
             $T(\ssx)$, shown as an interval. The corresponding prism $\pK$ on
             the bottom.  A 1-cycle in $K[11,17]$ and its lift in $\pK[11,17]$ are
             highlighted in red. The lifted cycle is an apex representative; its
             slices (pairs of vertices in red) are representatives of a bar in the original zigzag.}
    \label{fig:lifted-cycle}
\end{figure}

\section{Introduction}

In topological data analysis, one often wants to not only describe the
distribution of topological features in data, but also to identify individual
features directly in the input. For persistent homology, a natural starting point
are the cycle representatives of homology classes. Specifically, one wants to
find \emph{barcode representatives}, a consistent choice of bases for the
input sequence of homology groups.  For ordinary persistence, this comes
directly from the computation: a representative at the start of the bar includes
into every consecutive space and therefore remains a valid choice of the basis
element.

For zigzag persistence, the situation is more complicated. Because simplices can
enter and leave the complex, a representative cycle at one step in the zigzag
filtration may not exist at a later step. So even though zigzag persistence
algorithms keep track of representative cycles, the specific choices that they
make do not always produce compatible barcode representatives. To cope with this
problem, Dey et al.~\cite{DHM24} introduced an algorithm that keeps track of
additional information that allows it to report
consistent barcode representatives. For a zigzag filtration with $m$ insertions
and deletions, the algorithm
requires $O(m^3)$ preprocessing and $O(m^2)$ time to report one or all
representatives for any single bar, and $O(m^3)$ time to report all representatives.

On the surface, this is as well as one can hope to do: there are $O(m)$ bars,
over $O(m)$ spaces, with each representative consisting of $O(m)$ simplices.
But we can do better.
Every zigzag
filtration is isomorphic to a levelset zigzag filtration of some function.
The function generates the Mayer--Vietoris pyramid~\cite{CdSM09}, which
decomposes into diamond summands~\cite{BEMP13}. Each diamond has a single space
as its origin, its apex: this space maps into every other space in the diamond.
By finding the diamond's representative in
the apex --- an apex representative --- we can quickly recover its
representative in every other space, including those that fall in the levelset
zigzag, and therefore in the original zigzag; see \cref{fig:lifted-cycle}. All the information can be
recovered from an ordinary persistence computation. It takes
matrix multiplication time, $O(m^\omega)$, to compute persistence~\cite{MS24}, and an extra
$O(p \cdot m \log m)$ time to recover a $p$-dimensional
apex representative of size $O(p \cdot m)$. After the preprocessing, zigzag representatives
can be recovered in $O(\log m + C)$ time, where $C$ is the size of the output.

\section{Background}

We assume familiarity with simplicial and cell complexes, homology,
persistent homology, including zigzag persistence.
We only recap select topics to establish notation. Any
gaps can be filled with the standard literature~\cite{Hat02,EdHa10,EdMo17}.

Throughout the paper, we take care to work with arbitrary field coefficients,
but abuse notation as follows. Cell $\ssx \in z$ means the coefficient of $\ssx$ in
$z$ is non-zero, $\coeff{\ssx}{z} \neq 0$.
If $\alpha$ is a chain in some cell complex $K$, then $\alpha \subseteq L$ means that
$\alpha$ is supported on the subcomplex $L \subseteq K$, i.e., 
$\forall \ssx \in (K - L), \coeff{\ssx}{\alpha} = 0$.
Similarly, $z \cap L$ refers to the restriction of chain $z$ to $L$, i.e.,
$z \cap L = \sum_{\ssx \in L} \coeff{\ssx}{z} \cdot \ssx$.

\paragraph{Persistence.}
Given a filtered cell complex $K$, let $D$ be a matrix that represents its
boundary operator, with rows and columns ordered according to the filtration.
To find ordinary persistence pairing, one can compute a decomposition, $R =
DV$, where matrix $R$ is reduced, meaning its pivots --- lowest non-zero
elements in its columns --- appear in unique rows, and
matrix $V$ is full-rank upper-triangular. We index the columns and rows of the
matrices by the (totally ordered) cells of the input complex.
The persistence pairing is given by
the pivots of matrix $R$: a $p$-dimensional homology class created by the
addition of simplex $\ssx$ is destroyed by the addition of simplex $\tsx$ iff
$\low R[\tsx] = \ssx$, where $\low R[\tsx]$ returns the pivot in the column
$R[\tsx]$. By definition, the corresponding column $V[\tsx]$ stores a
$(p+1)$-dimensional chain whose boundary is the cycle $R[\tsx]$.

The decomposition $R=DV$ is not unique: multiple matrices $R$ and $V$ satisfy
it. However, the original persistence algorithm~\cite{ELZ02} performs operations
in a lazy fashion, subtracting columns from left to right only if their pivots
collide. We call this a \emph{lazy reduction}. It has the following special
properties that simplify our algorithms.

\begin{lemma}[Lemma 1 in \cite{NM24}]
    \label{lem:lazy-reduction}
    Assume decomposition $R = DV$ is obtained via the lazy reduction,
    and $\ssx_i = \low R[\tsx_i]$ and $\ssx_j = \low R[\tsx_j]$
    are such that $\tsx_i < \tsx_j$. If $\ssx_i < \ssx_j$,
    then entry $V[\tsx_i,\tsx_j] = 0$.
\end{lemma}
\begin{proof}
    Induction on $V$; see \cite{NM24}.
\end{proof}

\begin{remark}
    The authors of \cite{NM24} simplify the statement of the preceding lemma by
    abusing the notation. They assume that
    if $R[\tsx] = 0$, then $\low R[\tsx] = \bar{\ssx}$, a special imaginary cell
    that precedes every cell in the complex. In particular, it follows that
    in a lazy reduction, entry $V[\tsx_i, \tsx_j] = 0$ if $R[\tsx_i] = 0$.
\end{remark}

The contrapositive of the lemma is the following corollary.
\begin{corollary}
    \label{cor:lazy-nested}
    Assume decomposition $R = DV$ is obtained via the lazy reduction,
    and $\ssx_i = \low R[\tsx_i]$ and $\ssx_j = \low R[\tsx_j]$ are such that
    $\tsx_i < \tsx_j$. If
    $V[\tsx_i, \tsx_j] \neq 0$, then $\ssx_j < \ssx_i < \tsx_i < \tsx_j$.
\end{corollary}

\paragraph{Zigzag persistence.}
We start with a zigzag of simplicial complexes,
\begin{equation}
    \label{eq:space-zigzag}
    K_0 \to K_1 \ot K_2 \to K_3 \ot K_4 \to \ldots \ot K_n,
\end{equation}
where every consecutive pair of complexes vary by at most one simplex, i.e.,
$K_i$ can be equal to $K_{i+1}$ in the sequence.
It is convenient to assume that every simplex $\ssx$ appears and disappears exactly
once in this sequence, and therefore is present during some interval $T(\ssx) =
[i,j]$, i.e., $\ssx \in K_k$ iff $k \in T(\ssx)$. This assumption can be
made without loss of generality by assuming that the union $K$ of all simplices
across all times, $K = \bigcup_i K_i$ forms a $\Delta$-complex.
The full
definition of this object is too verbose --- we refer the reader to \cite{DH22},
which adapts the classical concept of $\Delta$-complexes in
\cite[p.~103]{Hat02} to our case of zigzag filtration --- but informally it relaxes the requirement that in a
simplicial complex two simplices intersect in a common face. All the familiar
notions of simplicial homology work the same, but now instead of having a
simplex supported over multiple intervals in the zigzag, we can have multiple
simplices, with the same boundary, each supported on a single interval.
We use $m = |K|$ to denote the size of the input complex.

\begin{remark*}
To simplify the notation we let $\min(\ssx) = \min T(\ssx)$ and $\max(\ssx) = \max T(\ssx)$.
\end{remark*}

\begin{remark}
    \label{rmk:generic-time-intervals}
    We assume that $\max(\ssx) > \min(\ssx)$, i.e., $T(\ssx)$ is not a
    single point. We can assume this without loss of generality because we can
    always pad the zigzag with extra copies of a space.
\end{remark}

Applying homology, we get a zigzag of homology groups, where we suppress
dimension for simplicity,
\begin{equation}
    \label{eq:zigzag}
    \Hgr(K_0) \to \Hgr(K_1) \ot \Hgr(K_2) \to \Hgr(K_3) \ot \Hgr(K_4) \to \ldots \ot \Hgr(K_n).
\end{equation}
Just like with ordinary persistence, this zigzag decomposes into $k$ bars
or persistence intervals~\cite{CdS10}.
Let $\{ \alpha^1_i, \ldots, \alpha^k_i \}$ be a choice of elements in $\Hgr(K_i)$
such that the non-zero elements form a basis for $\Hgr(K_i)$. We say that
such bases are \emph{compatible} across the zigzag if the maps in \cref{eq:zigzag} diagonalize, i.e.,
$\alpha^j_{2i-1} \mapsot \alpha^j_{2i} \mapsto \alpha^j_{2i+1}$. For every $j$,
$\alpha^j_*$ are non-zero over a single persistence interval $[b^j,d^j]$.
Cycles $z_i^j$ that give a set of compatible bases
$\alpha_i^j = [z_i^j]$ are called \emph{zigzag representatives}, with $z_*^j$
being the representatives of the $j$-th bar in the zigzag barcode.

\paragraph{Real-valued function.}
A convenient setting for persistence is that of a Morse-like~\cite{CdSKM19} real-valued
function $f: \Xsp \to \Rsp$.
We denote with $\Xsp_a^b = f^{-1}[a,b]$ the preimage of an interval and allow
the endpoints to be infinite, $a,b = \pm\infty$, in which case the
interval is understood to be open at the infinite ends. We denote the following pairs of spaces:
\begin{align*}
    & \Xsp[b,d] = (\Xsp_b^d, \emptyset),
   && \Xsp(b,d] = (\Xsp_{-\infty}^d, \Xsp_{-\infty}^b), \\
    & \Xsp[b,d) = (\Xsp_b^{\infty}, \Xsp_d^{\infty}),
   && \Xsp(b,d) = (\Xsp, \Xsp_{-\infty}^b \cup \Xsp_d^{\infty}).
\end{align*}

Let $a_1 < \ldots < a_n$ be the critical values of the function $f$.
Let $s_i$ be regular values interleaved with the critical values: $s_0 < a_1 <
s_1 < a_2 < \ldots < s_{n-1} < a_n < s_n$.
The following constructions play an important role in the
theory of persistent homology:
\begin{itemize}[leftmargin=*]
    \item
        Extended persistence (EP):
        \begin{align*}
            0& \to \Hgr(\Xsp[-\infty, s_0]) \to \ldots \to \Hgr(\Xsp[-\infty,
            s_n]) = \Hgr(\Xsp) \\
             & \to \Hgr(\Xsp[-\infty, s_n)) \to \ldots \to \Hgr(\Xsp[-\infty, s_1))
               \to \Hgr(\Xsp[-\infty, s_0)) = 0.
        \end{align*}
    \item
        Levelset zigzag (LZZ):
        \[
            0 \ot \Hgr(\Xsp[s_0, s_0]) \to \Hgr(\Xsp[s_0, s_1]) \ot
            \Hgr(\Xsp[s_1,s_1]) \to \ldots \to \Hgr(\Xsp[s_{n-1}, s_n]) \ot
            \Hgr(\Xsp[s_n,s_n]) \to 0.
        \]
\end{itemize}

Carlsson et al.~\cite{CdSM09} showed that these two sequences contain the same
information by arranging the four types of spaces into a \emph{Mayer--Vietoris
Pyramid}, see \cref{fig:mayer-vietoris-pyramid}(left). Once homology is applied, the
pyramid unrolls into an infinite Mayer--Vietoris strip of homology groups,
where every diamond belongs to the Mayer--Vietoris long exact sequence,
making it \emph{exact} in the terminology of \cite{CdS10}, see
\cref{fig:mayer-vietoris-pyramid}(right).
This allows one to translate the decomposition between any two paths that differ
by a single diamond (and therefore between any pair of paths via composition).

\begin{figure}
    \centering
    \includegraphics{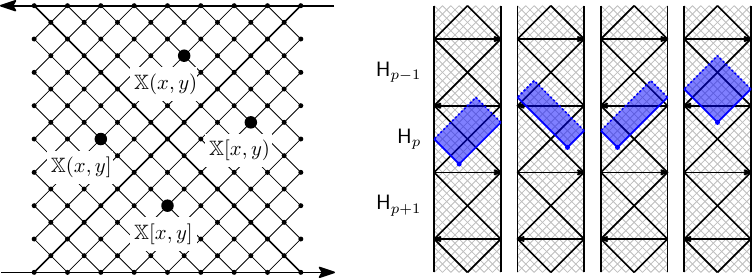}
    \caption{Left: Mayer--Vietoris Pyramid arranges spaces $\Xsp[b,d], \Xsp(b,d],
             \Xsp[b,d), \Xsp(b,d)$ in a way where every diamond in the pyramid belongs to
             the Mayer--Vietoris long exact sequence. Right: Four types of flush
             diamond indecomposables in the infinite strip of homology groups;
             the apex of each diamond is marked.}
    \label{fig:mayer-vietoris-pyramid}
\end{figure}

Bendich et al.~\cite{BEMP13} further showed that the translation rules in
\cite{CdSM09} hold at the level of the basis elements of the individual paths,
which means the entire pyramid decomposes as a direct sum of
(indecomposable) pointwise $1$-dimensional diamond summands, flush with its boundaries, shown in blue in \cref{fig:mayer-vietoris-pyramid}(right). The decomposition
of the levelset zigzag and extended persistence are simply slices through these
flush diamond summands.
Bauer et al.~\cite{BBF21} study the Mayer--Vietoris Pyramid in the cohomological setting.

We call the bottom-most space in the support of a diamond its \emph{apex}. The
homology class assigned to the diamond in that space, an \emph{apex class}, and
any cycle that belongs to such a class, an \emph{apex representative}.
We note that just like the choice of the basis for ordinary persistence is not
unique, neither is the choice of the apex classes.
An apex class is characterized by lying outside the image of the maps into the
apex; specifically, for the four types of apexes their classes satisfy:

\vspace{1ex}
\hspace{-5.5ex}
\begin{tabular}{l@{~}c@{~}l@{~}c@{~}l}
        $[z] \in \Hgr(\Xsp[b,d])$
         & s.t. &
        $[z] \notin \im (\Hgr(\Xsp[b+\ee,d]) \to \Hgr(\Xsp[b,d]))$
         & and &
        $[z] \notin \im (\Hgr(\Xsp[b,d-\ee]) \to \Hgr(\Xsp[b,d]))$ \\
        $[z] \in \Hgr(\Xsp(b,d])$
         & s.t. &
        $[z] \notin \im (\Hgr(\Xsp(b-\ee,d]) \to \Hgr(\Xsp(b,d]))$
         & and &
        $[z] \notin \im (\Hgr(\Xsp(b,d-\ee]) \to \Hgr(\Xsp(b,d]))$ \\
        $[z] \in \Hgr(\Xsp[b,d))$
         & s.t. &
        $[z] \notin \im (\Hgr(\Xsp[b+\ee,d)) \to \Hgr(\Xsp[b,d)))$
         & and &
        $[z] \notin \im (\Hgr(\Xsp[b,d+\ee)) \to \Hgr(\Xsp[b,d)))$ \\
        $[z] \in \Hgr(\Xsp(b,d))$
         & s.t. &
        $[z] \notin \im (\Hgr(\Xsp(b+\ee,d)) \to \Hgr(\Xsp(b,d)))$
         & and &
        $[z] \notin \im (\Hgr(\Xsp(b,d-\ee)) \to \Hgr(\Xsp(b,d)))$
\end{tabular}
\vspace{1ex}

Because
there is a map from the apex to every space in the diamond, we
can recover a representative in any space by simply mapping an apex
representative forward. All the maps in the pyramid are inclusions, except for
the boundary homomorphism connecting two consecutive dimensions of homology in
the strip. We say more about this translation in \cref{sec:zigzag-representatives}, but the
overarching point is that given an apex representative, it is straightforward to
recover every other representative, and levelset zigzag representatives in
particular.

\section{Setup}
\label{sec:setup}

\paragraph{Prism.}
Given a zigzag in \cref{eq:space-zigzag},
we define a larger cell complex $\pK$, called a \emph{prism}:
\begin{equation}
    \label{eq:prism}
    \pK = \left\{ \ssx \times [i,i+1] \mid [i,i+1] \subseteq T(\ssx), i \in \Zsp \right\}
                \; \cup \;
          \left\{ \ssx \times i \mid i \in T(\ssx), i \in \Zsp \right\},
\end{equation}
together with a function $\pf: |\pK| \to \Rsp$ on its underlying space,
defined by the projection onto the second component, $\pf(x \times t) = t$.
See \cref{fig:prism}.
We observe that the prism consists of two types of cells:
\emph{horizontal cells} $\ssx \times [i, i+1]$ and \emph{vertical cells} $\ssx \times i$.

\begin{figure}
    \centering
    \includegraphics{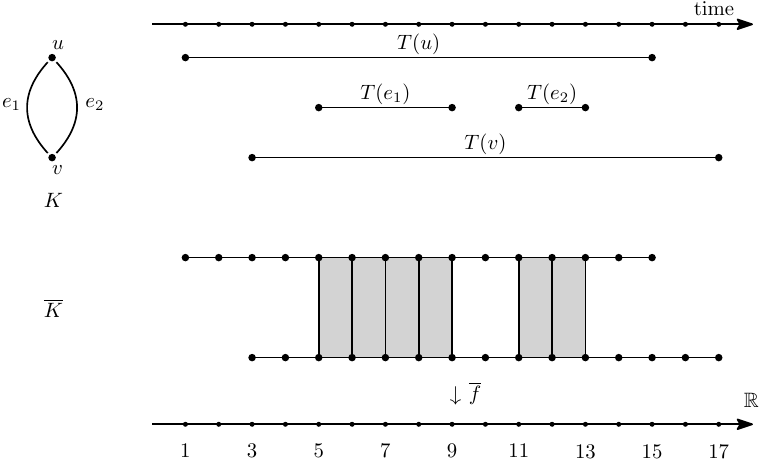}
    \caption{A $\Delta$-complex $K$, with the times $T(\ssx)$ of its simplices
             illustrated, together with the corresponding prism $\pK$ and its
             projection $\pf$ onto the second coordinate.}
    \label{fig:prism}
\end{figure}

We define an integer interlevel set, for $i,j \in \Zsp$,
$
    \pK_i^j = \{ \ssx \times T \in \pK \mid T \subseteq [i, j] \}.
$
We note that its underlying space $|\pK_i^j| = \pf^{-1}[i,j]$.
Then for any pair of integers $b,d$, we get four types of pairs of subspaces,

\hspace{-4ex}
\begin{tabular}{llll}
    $\pK[b,d] = \pK_b^d$,                    &
    $\pK(b,d] = (\pK_{-1}^d, \pK_{-1}^b)$,   &
    $\pK[b,d) = (\pK_b^{n+1}, \pK_d^{n+1})$, &
    $\pK(b,d) = (\pK, \pK_{-1}^b \cup \pK_d^{n+1})$.
\end{tabular}
Applying homology to each pair, we get relative homology groups that
appear in the four quadrants of the Mayer--Vietoris pyramid:
$\Hgr(\pK[b,d]), \Hgr(\pK(b,d])), \Hgr(\pK[b,d)), \Hgr(\pK(b,d))$.

Because even-indexed spaces in \cref{eq:space-zigzag} include into odd-indexed
spaces,
$K_{2i-1} \ot K_{2i} \to K_{2i+1}$, we have an isomorophism of
homology groups $\Hgr(\pK[2i,2i+1]) \simeq \Hgr(\pK[2i+1,2i+1]) \simeq
\Hgr(\pK[2i+1,2i+2])$, induced by deformation retractions.
It follows that the levelset zigzag~\cite{CdSM09} of function $\pf$,
\[
    \underbrace{\Hgr(\pK[0,0])}_{\Hgr(K_0)} \to
    \underbrace{\Hgr(\pK[0,1]) \stackrel{\simeq}{\ot} \Hgr(\pK[1,1]) \stackrel{\simeq}{\to} \Hgr(\pK[1,2])}_{\Hgr(K_1)} \ot
    \underbrace{\Hgr(\pK[2,2])}_{\Hgr(K_2)} \to
    \ldots \ot
    \underbrace{\Hgr(\pK[n,n])}_{\Hgr(K_n)}
\]
is isomorphic to our starting zigzag in \cref{eq:zigzag}.
Explicitly, the isomorphism is given by the following bijection between the four
types of intervals:
\begin{align*}
    & (\Hgr(K_{2i}),   \Hgr(K_{2j}))   & \leftrightarrow &&& (\Hgr(\pK[2i,2i]),   \Hgr(\pK[2j,2j]))         && \text{(open-open)}\\
    & (\Hgr(K_{2i}),   \Hgr(K_{2j+1})) & \leftrightarrow &&& (\Hgr(\pK[2i,2i]),   \Hgr(\pK[2j+1,2j+2]))     && \text{(open-closed)}\\
    & (\Hgr(K_{2i+1}), \Hgr(K_{2j}))   & \leftrightarrow &&& (\Hgr(\pK[2i,2i+1]), \Hgr(\pK[2j,2j]))         && \text{(closed-open)}\\
    & (\Hgr(K_{2i+1}), \Hgr(K_{2j+1})) & \leftrightarrow &&& (\Hgr(\pK[2i,2i+1]), \Hgr(\pK[2j+1,2j+2]))     && \text{(closed-closed)}
\end{align*}

\paragraph{Intervals.}

Because we perform computation on the input complex $K$, rather than
the prism, we need a notation for different intervals and pairs of intervals.
The connection between these and the corresponding intervals in $\pK$ will be
made clear in \cref{sec:intervals}.
We denote the sub- and super-level sets:
\begin{align*}
  & \subK{i} = \{ \ssx \in K \mid \min(\ssx) \leq i \}, &&
    \supK{j} = \{ \ssx \in K \mid \max(\ssx) \geq j \},
\end{align*}
and define four pairs of spaces:
\begin{align*}
    & K[b,d] = \supK{b} \cap \subK{d}, &
    & K(b,d] = (\subK{d}, \subK{b}), \\
    & K[b,d) = (\supK{b}, \supK{d}), &
    & K(b,d) = (K, \subK{b} \cup \supK{d}).
\end{align*}
Because every step of the input zigzag is a simplicial complex, it follows that
for any $\ssx,\tsx \in K$,
\begin{align}
    \label{eq:face-nests}
    \ssx \in \bdry \tsx \Rightarrow \min(\ssx) < \min(\tsx) < \max(\tsx) < \max(\ssx).
\end{align}
The following relationships between sub- and super-levelsets follow immediately:
\begin{align}
    \alpha \subseteq \subK{x} & \Rightarrow \bdry \alpha \subseteq \subK{x} \label{eq:bdry-sub} \\
    \alpha \subseteq K - \supK{x} & \Rightarrow \alpha \subseteq \subK{x} \label{eq:sup-sub}
\end{align}

\paragraph{Cone.}
To compute the persistence decomposition of the zigzag, we follow the algorithm
of Dey and Hou~\cite{DH22}, who use a construction similar to extended
persistence~\cite{CEH09} on the cone $\cK = \omega * K$ over the input complex $K$,
where $*$ indicates a join with the cone vertex $\omega$.
Every simplex $\ssx$ in the base space $K$ gets value $\min(\ssx)$ from the
input zigzag. Every simplex $\cssx$ in the cone $\cK - K$ gets value
$\max(\ssx)$ from the input zigzag. We then define a filtration, where the base
space simplices come first, ordered by increasing value, and the cone simplices
come second, ordered by decreasing value.  Specifically,
\[
    0 \to \Hgr(\subK{1}) \to \ldots \to \Hgr(\subK{n}) = \Hgr(K) \to
    \Hgr(K, \supK{n}) \to \ldots \to \Hgr(K, \supK{1})
    \to \Hgr(K,K) = 0,
\]
which in reduced homology is isomorphic to
\begin{equation}
    \label{eq:cone-filtration}
    0 \to \Hgr(\subK{1}) \to \ldots \to \Hgr(K) \to
    \Hgr(K \cup \omega * \supK{n}) \to \ldots \to
    \Hgr(K \cup \omega * \supK{1}) \to \Hgr(\cK) = 0.
\end{equation}
The latter is an ordinary filtration and we compute its persistence via
the $R = DV$ decomposition of the boundary map of the cone $\cK$.
Crucially, the resulting bars are in one-to-one correspondence with the bars in
the decomposition of the input zigzag in \cref{eq:zigzag}.

\section{Lifting}
\label{sec:lifting}

Because prism $\pK$ has so many cells, it is expensive to process directly. We want to
perform computation on the cone $\cK$ instead. To recover the apex representatives for
the prism, we need to lift cone cycles in $\cK$ to prism cycles in $\pK$. In \cref{sec:intervals}, we
explain what exact cycles to lift to get different apex representatives.
Meanwhile, we first describe an algorithm that given a (relative) cycle $z$ in
$KI$ produces a cycle $\pz$ in $\pK I$, where $I$ is one of the four types of
intervals, $(b,d]$, $[b,d)$, $(b,d)$, $[b,d]$.
\cref{alg:lift-cycle-reinterpreted} in this section is a reinterpretation of an
inefficient, but easier-to-understand \cref{alg:lift-cycle}, relegated to
\cref{apx:lifting} for lack of space. We encourage the reader to go through
\cref{alg:lift-cycle} first.

\begin{remark*}
    A $p$-chain $\pz$ in $\pK$ consists of two types of $p$-cells: $\tsx \times t$,
    where $\tsx$ is a $p$-simplex, and $\ssx \times [t_1, t_2]$, where $\ssx$ is a
    $(p-1)$-simplex.
\end{remark*}
\begin{remark*}
    We use an abridged notation,
    $c \cdot \ssx \times [i,j]$, to represent the chain
    $\sum_{k \in [i, j-1]} c \cdot \ssx \times [k,k+1]$.
    This choice is guided by computational efficiency:
    the former requires constant space vs.\ the $O(j-i)$ space required for the latter.
\end{remark*}

\cref{alg:lift-cycle-reinterpreted} takes a $p$-cycle $z$, a direction expressed as a pair of
values $(s,f)$, and an initial $(p-1)$-cycle $\winit$.
We call the pair $(s,f)$ a direction because the order of the values specifies
whether we process the cycle in increasing ($s < f$) or decreasing ($f < s$)
order. The algorithm stretches the cycle from $\pK[s,s]$ to $\pK[f,f]$,
covering $\pK I$; see \cref{fig:lifted-cycle}.

\begin{algorithm}
    \caption{\algname{Lift-Cycle}($z$, $(s,f)$, $\winit$) (Reinterpretation of
    \cref{alg:lift-cycle})} %
    \begin{algorithmic}[1]
        \Statex \textbf{Input:}
        \begin{itemize}[label=--, leftmargin=*]
            \item relative $p$-cycle $z$ in $H(KI)$, where $I$ is one of $(b,d]$, $[b,d)$, $(b,d)$, $[b,d]$;
            \item direction from start to finish, $s, f$ ($=$ $b,d$ or $d,b$);
            \item initial boundary $\winit = - \bdry z \cap K[s,s] \subseteq K[s,s]$
        \end{itemize}

        \Statex \textbf{Output:}
        \begin{itemize}[label=--, leftmargin=*]
            \item relative $p$-cycle $\pz$ in $H(\pK I)$
        \end{itemize}

        \For{\textit{each} $\tsx \in z$}
            \State $t_\tsx \gets \textrm{a value in}~ T(\tsx) \cap [b,d]$
            \Comment{if $T(\tsx) \cap [b,d] = \emptyset$, ignore $\tsx$} \;
            \State{$\pz \gets \pz + \coeff{\tsx}{z} \cdot (\tsx \times t_\tsx) $} \;
        \EndFor
        \For{\textit{each} $\ssx$ s.t.\ $\exists \tsx \in z, \ssx \in \bdry \tsx$} \Comment{each face of some simplex in $z$}
            \State $c \gets \coeff{\ssx}{\winit}$ \label{line:init-coeff} \;
            \State $l \gets s$ \;
            \For{$(t_\tsx,\tsx) \in \left\{ (t_\tsx, \tsx) \mid \tsx \in \cbdry \ssx \right\}$ in order of $t_\tsx$ from $s$ to $f$}
                \If{$l \neq t_\tsx$}
                    \State $\pz \gets \pz + c \cdot (\ssx \times [l, t_\tsx])$ \label{line:z-update-1-r} \;
                \EndIf
                \State $c \gets c + \coeff{\tsx}{z} \cdot \coeff{\ssx}{\bdry \tsx}$ \;
                \State $l \gets t_\tsx$ \;
            \EndFor
            \If{$l \neq f$}
                \State $\pz \gets \pz + c \cdot (\ssx \times [l,f]) \label{line:z-update-2-r}$
            \EndIf
        \EndFor
    \end{algorithmic}
    \label{alg:lift-cycle-reinterpreted}
\end{algorithm}

\paragraph{Correctness.}
The correctness of \cref{alg:lift-cycle-reinterpreted} follows from the
following claim about the boundary structure, which will be important in
\cref{sec:intervals}.
We let $z_x^y$ be the restriction of cycle $z$ to the simplices $\tsx$ whose times
$t_\tsx$ lie in the interval $[x,y]$, i.e.,
$z_x^y = \sum_{\tsx \in z, t_\tsx \in [x,y]} \coeff{\tsx}{z} \cdot \tsx$.

\begin{claim}
    \label{clm:lifted-boundary}
    Given input cycle $z \in K$, the boundary of the lifted cycle $\pz$ satisfies
    $
        \bdry \pz = \winit \times s + \wfinal \times f,
    $
    where
    $
        \wfinal = \winit + \bdry z_b^d.
    $
\end{claim}
\begin{proof}
    Simplex $\ssx$ generates a cell $\ssx \times [s,l] \in \pz$ iff
    $\coeff{\ssx}{\winit} \neq 0$ (Line \ref{line:init-coeff}). This cell contributes
    $\ssx \times s$ to the boundary of $\pz$.
    Simplex $\ssx$ generates a cell $\ssx \times [l',f] \in \pz$ iff
    $\coeff{\ssx}{\winit} + \sum_{\tsx \in \cbdry \ssx \cap z_b^d} \coeff{\ssx}{\tsx} \neq 0$.
    By definition this is the case only for simplices in $\wfinal = \winit +
    \bdry z_b^d$.
\end{proof}
It follows that if $\winit \subseteq K[s,s]$, then $\pz$ is a cycle in $\pK I$.
What is not immediate is why the chains $\ssx \times [l,t_\tsx]$ and
$\ssx \times [l,f]$ in
Lines \ref{line:z-update-1-r} and \ref{line:z-update-2-r} are in the prism $\pK$.

\begin{claim}
    Chains $\ssx \times [l, t_\tsx]$ and $\ssx \times [l, f]$ added in Lines
    \ref{line:z-update-1-r} and \ref{line:z-update-2-r} of
    \cref{alg:lift-cycle-reinterpreted} are present in $\pK$.
\end{claim}
\begin{proof}
    Without loss of generality, assume $s < f$. There are three types of chains:
    \begin{enumerate}
        \item
            $\ssx \times [s,t_{\tsx_1}]$.
            $\ssx$ has a non-zero coefficient on this interval iff
            $\coeff{\ssx}{\winit} \neq 0$ (Line \ref{line:init-coeff}), in which case $\ssx \in K[s,s]$.
            Existence of coface $\tsx_1$ implies
            $\min(\ssx) \leq s \leq t_{\tsx_1} < \max(\ssx)$.
        \item
            $\ssx \times [t_{\tsx_i}, t_{\tsx_{i+1}}]$.
            From \cref{eq:face-nests},
            $\min(\ssx) < t_{\tsx_i} \leq t_{\tsx_{i+1}} < \max(\ssx)$.
        \item
            $\ssx \times [t_{\tsx_i}, f]$.
            Suppose that this chain is not in $\pK$.
            This means $x = \max(\ssx) \in [t_{\tsx_i},f)$. \cref{eq:face-nests}
            implies that $\ssx$ has no cofaces in $[x,f]$.
            Therefore, $\ssx \in \winit + \bdry z_s^x$ and
            $\ssx \notin \bdry z_x^f$. Therefore, $\ssx \in \bdry z$. Since we
            assumed $z$ is a relative cycle in $KI$, it cannot have any boundary
            inside the interval itself --- a contradiction.
            \qedhere
    \end{enumerate}
\end{proof}

\begin{remark*}
    It may seem strange that we state nothing about the relationship between $z$
    and $\pz$. This is because we exploit additional properties of the algorithm
    in \cref{sec:intervals}.
\end{remark*}

\paragraph{Running time.}
Let $m$ be the size of the input $p$-cycle $z$.
\cref{alg:lift-cycle-reinterpreted} requires $O(p \cdot m \log m)$ time.
It goes through all simplices $\ssx$ that are faces of some simplex $\tsx \in z$;
there are at most $(p+1) \cdot m$ such $(p-1)$-simplices. For each one, the
algorithm sorts its cofaces by their times $t_\tsx$.
Although there is no bound on the number of cofaces of one simplex, the total
number of the cofaces to sort is again at most $(p+1) \cdot m$.
Individual operations in the update take constant time, so the overall running
time of \cref{alg:lift-cycle-reinterpreted} is $O(p \cdot m \log m)$. The size of the lifted cycle is $O(p \cdot m)$.

\section{Apex Representatives}
\label{sec:apex-representatives}

We give a characterization of apex representatives in the prism.

\begin{lemma}[Closed endpoint]~
    \label{clm:closed-endpoint}
    \begin{enumerate}
            \item
            If a relative cycle $\pz$, with $[\pz] \in \Hgr(\pK(b,d])$, contains $\tsx \times d$ with $\min \tsx = d$, then
            \[
                [\pz] \not\in \im \left( \Hgr(\pK(b,d-1]) \to \Hgr(\pK(b,d]) \right).
            \]
        \item
            If a relative cycle $\pz$, with $[\pz] \in \Hgr(\pK[b,d))$, contains $\tsx \times b$ with $\max \tsx = b$, then
            \[
                [\pz] \not\in \im \left( \Hgr(\pK[b+1,d)) \to \Hgr(\pK[b,d)) \right).
            \]
        \item
            If an absolute cycle $\pz$, with $[\pz] \in \Hgr(\pK[b,d])$, contains $\ssx \times b$
            and $\tsx \times d$ with $\max \ssx = b$ and $\min \tsx = d$, then
            \[
                [\pz] \not\in \im \left( \Hgr(\pK[b+1,d]) \to \Hgr(\pK[b,d]) \right)
                \quad \text{and} \quad
                [\pz] \not\in \im \left( \Hgr(\pK[b,d-1]) \to \Hgr(\pK[b,d]) \right).
            \]
    \end{enumerate}
\end{lemma}
\begin{proof}
    All the statements follow the same proof, except for the endpoints of the
    intervals, so we only show the first.
    Suppose there is some class $[\alpha] \in \Hgr(\pK(b,d-1])$ that
    maps to $[\pz] \in \Hgr(\pK(b,d])$.
    Then $\pz = \alpha + \bdry \beta$ for some chain $\beta$.
    Since $\tsx \in \pz$ and $\tsx \not\in \alpha$ (since $\min \tsx = d > d-1$),
    $\tsx$ must be in $\bdry \beta$. But from \cref{eq:face-nests}, $\tsx$ has no cofaces
    in $\subK{d}$ --- a contradiction.
\end{proof}

\begin{lemma}[Open endpoint]~
    \label{clm:open-endpoint}
    \begin{enumerate}
        \item
            If the boundary $\bdry \pz$ of a relative cycle $\pz$, with $[\pz] \in \Hgr(\pK(b,d])$,
            contains $\ssx \times b$ with $\min \ssx = b$, then
            $
                [\pz] \not\in \im \left( \Hgr(\pK(b-1,d]) \to \Hgr(\pK(b,d]) \right).
            $
        \item
            If the boundary $\bdry \pz$ of a relative cycle $\pz$, with $[\pz] \in \Hgr(\pK[b,d))$
            contains $\ssx \times d$ with $\max \ssx = d$, then
            $
                [\pz] \not\in \im \left( \Hgr(\pK[b,d+1)) \to \Hgr(\pK[b,d)) \right).
            $
        \item
            If the boundary $\bdry \pz$ of a relative cycle $\pz$, with $[\pz] \in \Hgr(\pK(b,d))$,
            contains $\ssx \times b$ with $\min \ssx = b$ and $\tsx \times d$ with $\max \tsx = d$, then
            \[
                \pz \not\in \im \left( \Hgr(\pK(b-1,d)) \to \Hgr(\pK(b,d)) \right)
                \quad \text{and} \quad
                \pz \not\in \im \left( \Hgr(\pK(b,d+1)) \to \Hgr(\pK(b,d)) \right).
            \]
    \end{enumerate}
\end{lemma}
\begin{proof}
    All the statements follow the same proof, except for the endpoints of the
    intervals, so we only show the first.
    From the long exact sequence of the triple~\cite[p.~118]{Hat02}, $\subK{b-1} \subseteq \subK{b} \subseteq \subK{d}$,
    we know that the image in the claim is equal to the kernel of the map
    induced by the boundary,
    \[
        \im \left( \Hgr_*(\pK(b-1,d]) \stackrel{\phantom{i^*}}{\to} \Hgr_*(\pK(b,d]) \right)
            =
        \ker \left( \Hgr_*(\pK(b,d]) \stackrel{\bdry^*}{\to} \Hgr_{*-1}(\pK(b-1,b]) \right).
    \]
    For $[\pz]$ to be in the kernel, $\bdry \pz$ needs to be a relative boundary in
    $\pK(b-1,b]$, i.e., $\bdry \pz = \bdry \beta + \gamma$ for some $\beta
    \subseteq \subK{b}$ and $\gamma \subseteq \subK{b-1}$.
    Since $\ssx \in \bdry \pz$ and $\ssx \not\in \gamma$, $\ssx$ must be in $\bdry \beta$.
    But from \cref{eq:face-nests}, $\ssx$ has no cofaces in $\subK{b}$ --- a contradiction.
\end{proof}

Putting \cref{clm:closed-endpoint,clm:open-endpoint} together, we get the
following theorem.
\begin{theorem}~
    \label{thm:apex-representatives}

    \begin{enumerate}
        \item
            If a relative cycle $\pz$, with $[\pz] \in \Hgr(\pK(b,d])$, contains $\tsx \times d$ with $\min \tsx = d$,
            and its boundary $\bdry \pz$ contains $\ssx \times b$ with $\min \ssx = b$, then $\pz$ is an apex representative.
        \item
            If a relative cycle $\pz$, with $[\pz] \in \Hgr(\pK[b,d))$, contains $\tsx \times b$ with $\max \tsx = b$,
            and its boundary $\bdry \pz$ contains $\ssx \times d$ with $\max \ssx = d$, then $\pz$ is an apex representative.
        \item
            If an absolute cycle $\pz$, with $[\pz] \in \Hgr(\pK[b,d])$, contains $\ssx \times b$
            and $\tsx \times d$ with $\max \ssx = b$ and $\min \tsx = d$, then
            $\pz$ is an apex representative.
        \item
            If the boundary $\bdry \pz$ of a relative cycle $\pz$, with $[\pz] \in \Hgr(\pK(b,d))$,
            contains $\ssx \times b$ with $\min \ssx = b$ and $\tsx \times d$ with $\max \tsx = d$, then
            $\pz$ is an apex representative.
    \end{enumerate}
\end{theorem}

\section{Four Intervals}
\label{sec:intervals}

We assume $R=DV$ is the decomposition of the boundary matrix $D$ of the cone $\cK$
filtration in \cref{eq:cone-filtration}. We use $\ssx,\tsx$ to refer to
simplices in the base space $K$, and $\cssx,\ctsx$ for the simplices in the
cone $\cK - K$.

\paragraph{Summary.}
The main content of this section is summarized in \cref{tbl:chain-summary}:
lifting the stated cycles $z$ in $K$, derived from the lazy reduction, with the
given arguments to the algorithm \algname{Lift-Cycle}, produces apex
representatives in $\pK$.
The results of the section are more general, but also more verbose; they derive the expression for the
cycles for any reduction.
\begin{table}[h]
    \centering
    \ra{1.3}    %
    \begin{tabular}{@{}lcll@{\quad}cc}
        \toprule
        EP type          & LZZ type      & Apex     & $z$               & $(s,f)$ & $\winit$\\
        \midrule
        Ordinary         & closed-open   & $K(b,d]$ & $V[\tsx]$         & $(d,b)$ & $0$ \\
        \midrule
        Relative         & open-closed   & $K[b,d)$ & $R[\ctsx] \cap K$ & $(b,d)$ & $0$ \\
        \midrule
        Extended $(d,b)$ & open-open     & $K[b,d]$ & $R[\ctsx]$        & $(d,b)$ & $0$ \\
        \midrule
        Extended $(b,d)$ & closed-closed & $K(b,d)$ & $V[\ctsx] \cap K$ & $(b,d)$ & $-R[\ctsx]$ \\
        \bottomrule
    \end{tabular}
    \caption{Summary of the chains and arguments used in different cases of the
    \algname{Lift-Cycle} algorithm, derived from the lazy reduction.}
    \label{tbl:chain-summary}
\end{table}

\subsection{Ordinary (closed-open)}
\label{sec:ordinary}

Throughout this subsection, we assume the following \textbf{setting}:
\[
    \text{a pair $(\ssx,\tsx)$ in the cone filtration, born at $b$ and dying at $d$, i.e., $\min(\ssx) = b < d = \min(\tsx)$.}
\]

    \includegraphics[page=3]{figures/pyramid}
        \quad
    \includegraphics[page=3]{figures/representatives}

    \begin{claim}
        \label{clm:ordinary-structure}
        Let $z = V[\tsx]$. Then $[z] \in \Hgr(K(b,d])$.
        Furthermore, $\tsx \in z$ and $\ssx \in \bdry z = R[\tsx]$.
    \end{claim}
    \begin{proof}
        Since $V$ is upper-triangular with all diagonal entries non-zero,
        $\tsx$ is the latest simplex in $V[\tsx]$, i.e.,
        $\max \left\{ \min(\tsx') \mid \tsx' \in V[\tsx] \right\} = \min(\tsx) = d$.
        Therefore, $V[\tsx] \subseteq \subK{d}$.
        By definition, the boundary
        $\bdry V[\tsx] = R[\tsx]$. $\ssx$ is its latest simplex, i.e.,
        $\max \left\{ \min(\ssx') \mid \ssx' \in R[\tsx] \right\} = \min(\ssx) = b$.
        Therefore, $\bdry V[\tsx] \subseteq \subK{b}$.
        It follows that if $z = V[\tsx]$, then $[z] \in \Hgr(\subK{d}, \subK{b}) = \Hgr(K(b,d])$.
    \end{proof}

\begin{claim}%
    \label{clm:ordinary-correct}
    Relative cycle $\pz$ produced by the call $\algname{Lift-Cycle}(V[\tsx], (d,b), 0)$
    is an apex representative in $\Hgr(\pK(b,d])$.
\end{claim}
\begin{proof}
    Let $z = V[\tsx]$.
    \begin{enumerate}
        \item
            $\ssx \times b \in \bdry \pz$, $\min(\ssx) = b$.

            Because $\ssx \in \bdry z$ (\cref{clm:ordinary-structure}) and $\min(\ssx) = b$, its coboundary
            $\cbdry \ssx \subseteq \supK{b}$. Meanwhile, $z \subseteq \subK{d}$. Therefore,
            all cofaces of $\ssx$ in $z$ intersect $[b,d]$. Therefore, $\ssx \in
            \bdry \left( z \cap [b,d] \right)$. It follows from
            \cref{clm:lifted-boundary} that
            \[
                \ssx \times b \in \bdry \pz
                    = \wfinal \times b
                    = \bdry \left( z \cap [b,d] \right) \times b.
            \]

        \item
            $\tsx \times d \in \pz$, $\min(\tsx) = d$.

            This follows immediately since $\tsx \in z$ (\cref{clm:ordinary-structure}) and
            $T(\tsx) \cap [b,d] = \{ d \}$.
    \end{enumerate}
    It follows from \cref{thm:apex-representatives} that $\pz$ is
    an apex representative in $\Hgr(\pK(b,d])$.
\end{proof}

\subsection{Relative (open-closed)}
\label{sec:relative}

Throughout this subsection, we assume the following \textbf{setting}:
\[
    \text{a pair $(\cssx,\ctsx)$ in the cone filtration, born at $d$ and dying at $b$,
    i.e., $\max(\ssx) = d > b = \max(\tsx)$.}
\]

    \includegraphics[page=2]{figures/pyramid}
        \quad
    \includegraphics[page=2]{figures/representatives}

    \begin{claim}
        \label{clm:relative-structure}
        Let chain $z$ consist of the
        boundary of the cone simplices in $V[\ctsx]$ restricted to the base
        space,
        \[
            z = \bdry \left(V[\ctsx] \cap (\cK - K) \right) \cap K.
        \]
        Then $[z] \in \Hgr(K[b,d))$.
        Furthermore, $\tsx \in z$ and $\ssx \in \bdry z$.
    \end{claim}
    \begin{proof}
        Let $\gamma = V[\ctsx] \cap (\cK - K)$, i.e., $z = \left( \bdry \gamma \right) \cap K$.

        \begin{enumerate}
            \item
                $z \subseteq \supK{b}$, $\tsx \in z$.

                Because $\gamma$ consists of only cone simplices, every simplex
                $\tsx' \in z$ is in the boundary of some cone simplex
                $\ctsx'$. Because
                $\ctsx$ is the latest simplex in $V[\ctsx]$, we have that
                $\ctsx'$ is added before $\ctsx$. Therefore,
                $\max(\tsx') \geq \max(\tsx) = b$.
                Because $\ctsx$ is the only simplex in $\gamma$ with $\tsx$ in
                its boundary, $\tsx \in \bdry \gamma \cap K = z$.

            \item
                $\bdry z \subseteq \supK{d}$.

                Let
                \[
                    \alpha = \bdry z = - \bdry \left( \bdry \gamma \cap (\cK - K) \right)
                \]
                be the shared boundary between the base space and the
                cone parts of $\bdry \gamma$.
                We note that because cone simplices can only be in the boundary
                of cone simplices,
                $( \bdry V[\ctsx] ) \cap (\cK - K) = \bdry \gamma \cap (\cK - K)$.
                Therefore, all the cone simplices $\cssx' \in \bdry \gamma \cap (\cK - K))$
                are such that $\max(\ssx') \geq \max(\ssx) = d$. Therefore, any
                simplex in the shared boundary $\alpha$ is in $\supK{d}$.

                Putting the first two sub-claims together, $[z] \in \Hgr(\supK{b}, \supK{d}) = \Hgr(K[b,d))$.
            \item
                $\ssx \in \bdry z$.

                $(\cssx, \ctsx)$ is a persistence pair, therefore, $\cssx \in \bdry V[\ctsx]$.
                Because cone simplices can only be
                in the boundaries of cone simplices, $\cssx \in \bdry \gamma$.
                Because $\cssx$ is the only simplex in $\left( \bdry \gamma \right) \cap (\cK - K)$
                that has $\ssx$ in its boundary, it follows that
                \[
                    \ssx \in \bdry z
                        = \bdry \left( \left( \bdry \gamma \right) \cap K \right)
                        = \bdry \left( \left( \bdry \gamma \right) \cap (\cK - K) \right).
                \qedhere
                \]
        \end{enumerate}
    \end{proof}

    \begin{remark}[Lazy reduction]
        If $R = DV$ is obtained by lazy reduction, it follows from
        \cref{cor:lazy-nested} that $V[\ctsx]$ consists only of cone simplices.
        In other words, the chain in the prior claim simplifies to
        \[
            z = \left( \bdry V[\ctsx] \right) \cap K = R[\ctsx] \cap K.
        \]
    \end{remark}

\begin{claim}%
    \label{clm:relative-correct}
    Relative cycle $\pz$ produced by the call $\algname{Lift-Cycle}(R[\ctsx] \cap K, (b,d), 0)$
    is an apex representative in $\Hgr(\pK[b,d))$.
\end{claim}
\begin{proof}
    Let $z = R[\ctsx] \cap K$.
    \begin{enumerate}
        \item
            $\ssx \times d \in \bdry \pz$, $\max(\ssx) = d$.

            Because $\ssx \in \bdry z$ (\cref{clm:relative-structure}) and $\max(\ssx) = d$, its coboundary
            $\cbdry \ssx \subseteq \subK{d}$. Meanwhile, $z \subseteq \supK{b}$. Therefore,
            all cofaces of $\ssx$ in $z$ intersect $[b,d]$. Therefore, $\ssx \in
            \bdry \left( z \cap [b,d] \right)$. It follows from
            \cref{clm:lifted-boundary} that
            \[
                \ssx \times d \in \bdry \pz
                    = \wfinal \times d
                    = \bdry \left( z \cap [b,d] \right) \times d.
            \]

        \item
            $\tsx \times b \in \pz$, $\max(\tsx) = b$.

            This follows immediately since $\tsx \in z$ (\cref{clm:relative-structure}) and
            $T(\tsx) \cap [b,d] = \{ b \}$.
    \end{enumerate}
    It follows from \cref{thm:apex-representatives} that $\pz$ is
    an apex representative in $\Hgr(\pK[b,d))$.
\end{proof}

\subsection{Extended (open-open)}
\label{sec:extended-oo}

Throughout this subsection, we assume the following \textbf{setting}:
\[
    \text{a pair $(\ssx, \ctsx)$ in the cone filtration born at $d$, dying at $b$, i.e.,
    $\min(\ssx) = d > b = \max(\tsx)$.}
\]

    \includegraphics[page=1]{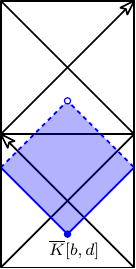}
        \quad
    \includegraphics[page=1]{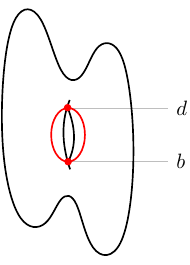}

    \begin{claim}
        \label{clm:extended-oo-structure}
        Let chain $c$ consist of those
        simplices in $V[\ctsx]$ that are either in the cone, or in $\supK{b}$,
        i.e.,
        \[
            c = V[\ctsx] \cap (\cK - K) + V[\ctsx] \cap \supK{b}.
        \]
        Let $z = \bdry c$; then $[z] \in \Hgr(K[b,d])$.
        Furthermore, $\ssx, \tsx \in z$.
    \end{claim}
    \begin{proof}
        We split the chain $V[\ctsx]$ into three parts:
        \begin{align*}
            & \alpha = V[\ctsx] \cap (K - \supK{b}), &
            & \beta  = V[\ctsx] \cap \supK{b}, &
            & \gamma = V[\ctsx] \cap (\cK - K),
        \end{align*}
        and re-write its boundary accordingly:
        $
            R[\ctsx] = \bdry V[\ctsx] = \bdry \alpha + \bdry \beta + \bdry \gamma.
        $
        $c = (\beta + \gamma) = V[\ctsx] - \alpha$.

        \begin{enumerate}
            \item
                $z = \bdry c \subseteq \subK{d}$, $\ssx \in z$.

                $\bdry V[\ctsx] \subseteq \subK{d}$ because $\ssx$ is the latest
                simplex in $R[\ctsx]$.
                Putting together \cref{eq:sup-sub,eq:bdry-sub},
                $\alpha \subseteq K - \supK{b}$ implies $\bdry \alpha \subseteq \subK{b} \subseteq \subK{d}$.
                Therefore, $z = \bdry V[\ctsx] - \bdry \alpha \subseteq \subK{d}$.

                Because $\ssx \in \bdry V[\ctsx]$, but $\ssx \not\in \bdry \alpha$,
                $\ssx \in z$.

            \item
                $z \subseteq \supK{b}$, $\tsx \in z$.

                Since chain $\beta \subseteq \supK{b}$, its boundary $\bdry \beta \subseteq \supK{b}$.
                Since boundaries $\bdry V[\ctsx], \bdry \alpha, \bdry \beta \subseteq K$, so is $\bdry \gamma$.
                Every simplex in $\gamma$ has the form $\omega * \tsx$, where
                $\tsx \in \supK{b}$. Therefore, $\bdry \gamma \subseteq \supK{b}$.
                It follows that $z = \bdry \beta + \bdry \gamma \subseteq \supK{b}$.

                Because $\max(\tsx) = b$, it has no cofaces in $\supK{b}$.
                Therefore, $\tsx \not\in \bdry \beta$. Because $\ctsx$ is the
                only simplex in $\gamma$ that has $\tsx$ in its boundary,
                $\tsx \in \bdry \gamma$. Therefore, $\tsx \in z = \bdry \beta + \bdry \gamma$.
        \end{enumerate}

        Therefore, $[z] \in \Hgr(\supK{b} \cap \subK{d}) = \Hgr(K[b,d])$.
    \end{proof}

    \begin{remark}[Lazy reduction]
        Furthermore, if $R = DV$ comes from the lazy reduction, then the next claim proves that
        in the previous claim, chain $c = V[\ctsx]$.
        In other words, we can use $z = R[\ctsx] = \bdry V[\ctsx]$ directly.
    \end{remark}

    \begin{claim}[Lazy reduction]
        \label{clm:lazy-V-sup}
        If $R=DV$ is obtained from the lazy reduction, then all the base space
        simplices in $V[\ctsx]$ are in $\supK{b}$, i.e., $V[\ctsx] \cap K \subseteq \supK{b}$.
    \end{claim}
    \begin{proof}
                Let $\tsx'$ be a simplex in $V[\ctsx] \cap K$. In this case,
                for the lazy reduction, because $V[\tsx', \ctsx] \neq 0$, from
                \cref{cor:lazy-nested}, taking
                $\ssx_j = \ssx$ and $\tsx_i = \tsx'$, we get
                \[
                    \max(\tsx') > \min(\tsx') > \min(\ssx) = d > b.
                    \qedhere
                \]
    \end{proof}

\begin{claim}%
    \label{clm:extended-oo-correct}
    Absolute cycle $\pz$ produced by the call $\algname{Lift-Cycle}(R[\ctsx], (d,b), 0)$
    is an apex representative in $\Hgr(\pK[b,d])$.
\end{claim}
\begin{proof}
    Let $z = R[\ctsx]$.
    Because $\ssx, \tsx \in z$ (\cref{clm:extended-oo-structure}) and
    $\min(\ssx) = d$ and $\max(\tsx) = b$,
    both simplices $\ssx \times d$ and $\tsx \times b$ are in the lifted cycle $\pz$.
    It follows from \cref{thm:apex-representatives} that $\pz$ is an apex
    representative in $\Hgr(\pK[b,d])$.
\end{proof}

\subsection{Extended (closed-closed)}
\label{sec:extended-cc}

Throughout this subsection, we assume the following \textbf{setting}:
\[
    \text{a pair $(\ssx,\ctsx)$ in the cone filtration, born at $b$, dying at $d$, i.e.,
    $\min(\ssx) = b < d = \max(\tsx)$.}
\]

    \includegraphics[page=4]{figures/pyramid}
        \quad
    \includegraphics[page=4]{figures/representatives}

    \begin{claim}
        \label{clm:extended-cc-structure}
        Let $z = V[\ctsx] \cap \supK{b}$. Then $[z] \in \Hgr(K(b,d))$.
        Furthermore,
        \[
            \ssx \in \winit = -\bdry \left( V[\ctsx] \cap (\cK - K) + V[\ctsx] \cap \supK{b} \right) \subseteq K[b,b]
        \]
        and $\tsx \in \left( \winit + \bdry z \right) = -\bdry \left( V[\ctsx] \cap (\cK - K) \right)$.
    \end{claim}
    \begin{proof}
        We split the chain $V[\ctsx]$ into three parts:
        \begin{align*}
            & \alpha = V[\ctsx] \cap (K - \supK{b}), &
            & z = \beta  = V[\ctsx] \cap \supK{b}, &
            & \gamma = V[\ctsx] \cap (\cK - K),
        \end{align*}
        and re-write its boundary accordingly:
        $
            R[\ctsx] = \bdry V[\ctsx] = \bdry \alpha + \bdry \beta + \bdry \gamma.
        $

        \begin{enumerate}
        \item
            $\bdry z \subseteq \subK{b} \cup \supK{d}$.

            Putting together \cref{eq:sup-sub,eq:bdry-sub},
            $\alpha \subseteq K - \supK{b}$ implies $\bdry \alpha \subseteq \subK{b}$.
            Because $\ssx$ is the latest simplex in $R[\ctsx]$ and $\min(\ssx) = b$,
            $R[\ctsx] \subseteq \subK{b}$.
            It follows that $\bdry (\beta + \gamma) = R[\ctsx] - \bdry \alpha \subseteq \subK{b}$.
            Since $\bdry \gamma \subseteq \supK{d}$, we have
            $\bdry z = \bdry (\beta + \gamma) - \bdry \gamma \subseteq \subK{b} \cup \supK{d}$.
            Therefore, $[z] \in \Hgr(K(b,d))$.

        \item
            $\ssx \in \winit = \bdry(\beta + \gamma) \subseteq K[b,b].$

            We already saw that $\bdry (\beta + \gamma) \subseteq \subK{b}$.
            To show that it is also in $\supK{b}$, we note that
            since $\beta \subseteq \supK{b}$, its boundary $\bdry \beta \subseteq \supK{b}$.
            $\bdry \gamma \subseteq \supK{d} \subseteq \supK{b}$.
            Therefore, $\bdry (\beta + \gamma) \subseteq \supK{b}$.

            Because $\min(\ssx) = b$, \cref{eq:face-nests} implies that $\ssx$
            cannot be in the boundary of any simplex $\tsx \in \alpha$ since
            $\max(\tsx) < b$. Since $\ssx \in R[\ctsx]$ and
            $\winit = \bdry(\beta + \gamma) = R[\ctsx] - \bdry \alpha$, it
            follows that $\ssx \in \winit$.

        \item
            $\tsx \in -\bdry \gamma.$

            Because $\ctsx$ is the only simplex in $\gamma$ that has $\tsx$ in
            its boundary, $\tsx \in -\bdry \gamma$.
            \qedhere
        \end{enumerate}
    \end{proof}

    \begin{claim}[Lazy reduction]
        If $R=DV$ is obtained from the lazy reduction, then
        $z = V[\ctsx] \cap K$ is in $K(b,d)$, $\ssx \in \winit = -R[\ctsx] = -\bdry V[\ctsx] \subseteq K[b,b]$,
        and $\tsx \in \left( \winit - \bdry z \right)$.
    \end{claim}
    \begin{proof}
        We only need to show that, in the case of the lazy reduction,
        $\alpha = V[\ctsx] \cap (K - \supK{b}) = 0$,
        i.e., $V[\ctsx] \cap K = V[\ctsx] \cap \supK{b}$.
        Then the claim follows from the previous \cref{clm:extended-cc-structure}.

        The proof is analogous to the proof of \cref{clm:lazy-V-sup}.
        Let $\tsx'$ be a simplex in $V[\ctsx] \cap K$. In this case,
        for the lazy reduction, because $V[\tsx', \ctsx] \neq 0$, from
        \cref{cor:lazy-nested}, taking
        $\ssx_j = \ssx$ and $\tsx_i = \tsx'$, we get
        $
            \max(\tsx') > \min(\tsx') > \min(\ssx) = b.
            \qedhere
        $
    \end{proof}

    \begin{remark}
        $V[\ctsx] \cap K$ can be empty.
        For example, if $K$ is a single vertex that appears at $b$ and then disappears at $d$.
        This is the reason why we need to bootstrap \cref{alg:lift-cycle,alg:lift-cycle-reinterpreted} with
        an initial cycle $\winit$.
    \end{remark}

\begin{claim}%
    \label{clm:extended-cc-correct}
    Relative cycle $\pz$ produced by the call $\algname{Lift-Cycle}(V[\ctsx] \cap K, (b,d), -R[\ctsx])$
    is an apex representative in $\Hgr(\pK(b,d))$.
\end{claim}
\begin{proof}~

    \begin{enumerate}
        \item
            $\ssx \times b \in \bdry \pz$, $\min(\ssx) = b$.

            Simplex $\ssx$, with $\min(\ssx) = b$, is in the initial cycle $\winit$ (\cref{clm:extended-cc-structure}).
            Therefore, $\ssx \times b$ is in the boundry $\bdry \pz$ of the lifted cycle.

        \item
            $\tsx \times d \in \bdry \pz$, $\max(\tsx) = d$.

            Simplex $\tsx$, with $\max(\tsx) = d$, is in $\winit + \bdry z$ (\cref{clm:extended-cc-structure}).
            Its coboundary $\cbdry \tsx \subseteq \subK{d}$.
            Meanwhile, $z \subseteq \supK{b}$ (\cref{clm:extended-cc-structure}).
            Therefore, all cofaces of $\tsx$ in $z$ intersect $[b,d]$.
            Therefore, $\tsx \in \winit + \bdry \left( z \cap [b,d] \right) = \wfinal$.
            It follows from \cref{clm:lifted-boundary} that
            \[
                \tsx \times d \in \wfinal \times d \subseteq \bdry \pz.
            \]
    \end{enumerate}
    It follows from \cref{thm:apex-representatives} that $\pz$ is an apex
    representative in $\Hgr(\pK(b,d))$.
\end{proof}

\section{Zigzag Representatives}
\label{sec:zigzag-representatives}

To solve our original problem --- to find zigzag representatives --- all that
remains is to map an apex representative into an appropriate space in the
zigzag.
Because we assumed in \cref{rmk:generic-time-intervals} that no
simplex is supported on a single space, the death and the birth for any given
interval are distinct.
Suppose we want to map an apex representative $\pz$, with $[\pz] \in \Hgr(\pK I)$,
into a zigzag representative in
$\Hgr(\pK[i,i]) \simeq \Hgr(K_i)$.

If $i$ is odd, then the $p$-cycle $\pz$ may contain a vertical cell $\tsx \times
i$, which does not map into $(p-1)$-cells in the levelset. To sidestep this
complication, we perturb $i$.
Let $x = i
\pm \ee \in I$ be a real value near $i$ that lies inside our birth-death
interval. We assume $x = i + \ee$; the other case is symmetric.
We implicitly subdivide $\pK$ at $x$ and extend our previous notation to allow
for the interlevel sets ending at $x$, e.g., $\pK_i^x$.
We note that $\pK_i^x$ deformation retracts onto $\pK_i^i$, with the homotopy
following the second (time) coordinate, and $\pK_x^x$ includes into $\pK_i^i$
once the cell times are shifted from $x$ to $i$.
We want to compose three maps:
\[
    \Hgr_p(\pK I) \to \Hgr_p(\pK(i,x)) \stackrel{\bdry^*}{\to} \Hgr_{p-1}(\pK[x,x]) \to \Hgr_{p-1}(K_i).
\]
Because $\pK I$ is an apex of a diamond, and $\pK(i,x)$ is a space in the diamond, the first map is an inclusion (of pairs), so $\pz$ remains a relative cycle in $\pK(i,x)$.
Let
\[
    \pz = \sum \alpha_\tsx \cdot (\tsx \times t_\tsx) + \sum \alpha_\ssx \cdot (\ssx \times [t_\ssx^1,t_\ssx^2]),
\]
then the second (boundary) map takes $\pz$ to
\[
    \pz[x,x] = \sum_{\substack{(\ssx \times [t_\ssx^1,t_\ssx^2]) \in \pz \\ x \in [t_\ssx^1,t_\ssx^2]}} \alpha_\ssx \cdot (\ssx \times x),
\]
which maps into $K_i$ by dropping the second coordinate,
\[
    z(i) = \sum_{(\ssx \times x) \in \pz[x,x]} \alpha_\ssx \cdot \ssx.
\]
By storing the apex cycle $\pz$ in an interval tree~\cite{Ede80,McC80}, we can
retrieve any zigzag representative of size $C$ in time $O(\log m + C)$.

\bibliographystyle{unsrturl}
\bibliography{refs}

\begin{thebibliography}{10}

\bibitem{DHM24}
Tamal~K. Dey, Tao Hou, and Dmitriy Morozov.
\newblock A fast algorithm for computing zigzag representatives.
\newblock In Yossi Azar and Debmalya Panigrahi, editors, {\em Proceedings of
  the 2025 Annual {ACM-SIAM} Symposium on Discrete Algorithms, {SODA} 2025, New
  Orleans, LA, USA, January 12-15, 2025}, pages 3530--3546. {SIAM}, 2025.
\newblock \href {https://doi.org/10.1137/1.9781611978322.116}
  {\path{doi:10.1137/1.9781611978322.116}}.

\bibitem{CdSM09}
Gunnar Carlsson, Vin de~Silva, and Dmitriy Morozov.
\newblock Zigzag persistent homology and real-valued functions.
\newblock In {\em Proceedings of the Twenty-fifth Annual Symposium on
  Computational Geometry}, SCG '09, pages 247--256, New York, NY, USA, 2009.
  ACM.
\newblock \href {https://doi.org/10.1145/1542362.1542408}
  {\path{doi:10.1145/1542362.1542408}}.

\bibitem{BEMP13}
Paul Bendich, Herbert Edelsbrunner, Dmitriy Morozov, and Amit Patel.
\newblock Homology and robustness of level and interlevel sets.
\newblock {\em Homology, Homotopy and Applications}, 15(1):51--72, 2013.

\bibitem{MS24}
Dmitriy Morozov and Primoz Skraba.
\newblock Persistent (co)homology in matrix multiplication time.
\newblock {\em arXiv [math.AT]}, December 2024.
\newblock \href {https://arxiv.org/abs/2412.02591} {\path{arXiv:2412.02591}}.

\bibitem{Hat02}
Allen Hatcher.
\newblock {\em Algebraic Topology}.
\newblock Cambridge University Press, 2002.

\bibitem{EdHa10}
Herbert Edelsbrunner and John Harer.
\newblock {\em Computational Topology: An Introduction}.
\newblock American Mathematical Soc., 2010.

\bibitem{EdMo17}
Herbert Edelsbrunner and Dmitriy Morozov.
\newblock Persistent homology.
\newblock In Jacob~E Goodman, Joseph O'Rourke, and Csaba~D Tóth, editors, {\em
  Handbook of Discrete and Computational Geometry}. CRC Press, 2017.
\newblock URL: \url{https://www.csun.edu/~ctoth/Handbook/chap24.pdf}.

\bibitem{ELZ02}
Herber Edelsbrunner, David Letscher, and Afra Zomorodian.
\newblock Topological persistence and simplification.
\newblock {\em Discrete \& computational geometry}, 28(4):511--533, November
  2002.
\newblock \href {https://doi.org/10.1007/s00454-002-2885-2}
  {\path{doi:10.1007/s00454-002-2885-2}}.

\bibitem{NM24}
Arnur Nigmetov and Dmitriy Morozov.
\newblock Topological optimization with big steps.
\newblock {\em Discrete \& computational geometry}, 72(1):310--344, July 2024.
\newblock \href {https://doi.org/10.1007/s00454-023-00613-x}
  {\path{doi:10.1007/s00454-023-00613-x}}.

\bibitem{DH22}
Tamal~K. Dey and Tao Hou.
\newblock {Fast Computation of Zigzag Persistence}.
\newblock In Shiri Chechik, Gonzalo Navarro, Eva Rotenberg, and Grzegorz
  Herman, editors, {\em 30th Annual European Symposium on Algorithms (ESA
  2022)}, volume 244 of {\em Leibniz International Proceedings in Informatics
  (LIPIcs)}, pages 43:1--43:15, Dagstuhl, Germany, 2022. Schloss Dagstuhl --
  Leibniz-Zentrum f{\"u}r Informatik.
\newblock \href {https://doi.org/10.4230/LIPIcs.ESA.2022.43}
  {\path{doi:10.4230/LIPIcs.ESA.2022.43}}.

\bibitem{CdS10}
Gunnar Carlsson and Vin de~Silva.
\newblock Zigzag persistence.
\newblock {\em Foundations of computational mathematics}, 10(4):367--405,
  August 2010.
\newblock \href {https://doi.org/10.1007/s10208-010-9066-0}
  {\path{doi:10.1007/s10208-010-9066-0}}.

\bibitem{CdSKM19}
Gunnar Carlsson, Vin de~Silva, Sara Kališnik, and Dmitriy Morozov.
\newblock Parametrized homology via zigzag persistence.
\newblock {\em Algebraic \& Geometric Topology}, 19(2):657--700, March 2019.
\newblock \href {https://doi.org/10.2140/agt.2019.19.657}
  {\path{doi:10.2140/agt.2019.19.657}}.

\bibitem{BBF21}
Ulrich Bauer, Magnus~Bakke Botnan, and Benedikt Fluhr.
\newblock Structure and interleavings of relative interlevel set cohomology.
\newblock {\em arXiv [math.AT]}, August 2021.
\newblock \href {https://arxiv.org/abs/2108.09298} {\path{arXiv:2108.09298}}.

\bibitem{CEH09}
David Cohen-Steiner, Herbert Edelsbrunner, and John Harer.
\newblock Extending persistence using {Poincaré} and {Lefschetz} duality.
\newblock {\em Foundations of computational mathematics}, 9(1):79--103,
  February 2009.
\newblock \href {https://doi.org/10.1007/s10208-008-9027-z}
  {\path{doi:10.1007/s10208-008-9027-z}}.

\bibitem{Ede80}
Herbert Edelsbrunner.
\newblock Dynamic data structures for orthogonal intersection queries.
\newblock Technical Report F59, Inst. Information Process., Graz University of
  Technology, Austria, 1980.
\newblock URL:
  \url{https://pub.ista.ac.at/~edels/Papers/1980-01-OrthogonalIntersectionQueries.pdf}.

\bibitem{McC80}
Edward~M. McCreight.
\newblock Efficient algorithms for enumerating intersecting intervals and
  rectangles.
\newblock Technical Report CSL-80-9, Xeox Palo Alto Research Center, 1980.

\end{thebibliography}

\appendix

\section{Lifting Algorithm}
\label{apx:lifting}

\cref{alg:lift-cycle} takes a $p$-cycle $z$, a direction expressed as a pair of
values $(s,f)$, and an initial $(p-1)$-cycle $\winit$.
We call the pair $(s,f)$ a direction because the order of the values specifies
whether we process the cycle in increasing ($s < f$) or decreasing ($f < s$)
order. The algorithm stretches the cycle from $\pK[s,s]$ to $\pK[f,f]$,
covering $\pK I$.

The algorithm maintains a $(p-1)$-cycle $w$ that gets stretched into the $p$-chain
$w \times [l,t_\tsx]$ at each step.
This cycle starts out as
$\winit \subseteq K[s,s]$, which is often, but not always 0.
The algorithm chooses for each simplex $\tsx \in z$ a time $t_\tsx$ that falls
inside the interval $I$ and processes the simplices according to these times
from $s$ to $f$. It adds each simplex as a vertical cell $\tsx \times t_\tsx$ to
the prism and updates $w$ with the boundary $\bdry \tsx$.
See \cref{fig:lifted-cycle} for an illustration of a cycle in $K$ and its lift
in $\pK$.

\begin{algorithm}
    \caption{\algname{Lift-Cycle-easy}($z$, $(s,f)$, $\winit$)}
    \begin{algorithmic}[1]
        \Statex \textbf{Input:}
        \begin{itemize}[label=--, leftmargin=*]
            \item relative $p$-cycle $z$ in $H(KI)$, where $I$ is one of $(b,d]$, $[b,d)$, $(b,d)$, $[b,d]$;
            \item direction from start to finish, $s, f$ ($=$ $b,d$ or $d,b$);
            \item initial boundary $\winit = \bdry z \cap K[s,s] \subseteq K[s,s]$
        \end{itemize}

        \Statex \textbf{Output:}
        \begin{itemize}[label=--, leftmargin=*]
            \item relative $p$-cycle $\pz$ in $H(\pK I)$
        \end{itemize}

        \For{\textit{each} $\tsx \in z$}
            \State $t_\tsx \gets \textrm{a value in}~ T(\tsx) \cap [b,d]$ \;
        \EndFor

        \State $\pz \gets 0$ \Comment{$p$-chain} \;
        \State $w \gets \winit$ \Comment{$(p-1)$-cycle} \;
        \State $l \gets s$ \Comment{``last'' time} \;
        \For{$(t_\tsx, \tsx)$ in order of $t_\tsx$ from $s$ to $f$}
            \State $a_\tsx = \coeff{\tsx}{z}$
            \State $\pz \gets \pz + a_\tsx \cdot (\tsx \times t_\tsx)$
                \label{line:z-update-1} \;
            \If{$l \neq t_\tsx$}
                \State $\pz \gets \pz + (w \times [l,t_\tsx])$
                \label{line:z-update-2} \;
            \EndIf
            \State $w \gets w + a_\tsx \cdot \bdry \tsx$ \label{line:w-update} \;
            \State $l \gets t_\tsx$ \;
        \EndFor
        \State $\wfinal \gets w$ \;
        \If{$l \neq f$}
            \State $\pz \gets \pz + (\wfinal \times [l, f])$
        \EndIf
    \end{algorithmic}
    \label{alg:lift-cycle}
\end{algorithm}

As stated, \cref{alg:lift-cycle} runs in quadratic time (two linear cycles $w
\times [l, t_\tsx]$ and $\pz$ get added at every step). It is possible to
improve this with a more careful implementation, by noticing that a simplex
$\ssx$ is output with the same coefficient at every step until this coefficient
changes when we encounter one of its cofaces.
Instead of introducing extra dictionaries for this bookkeeping,
\cref{alg:lift-cycle-reinterpreted} gives an alternative formulation of the same
algorithm. If one thinks of the update of the cycle $w$ in Line
\ref{line:w-update} of \cref{alg:lift-cycle}
as a for-loop over simplices $\ssx$ in cycles $w$ and $\bdry \tsx$, then
\cref{alg:lift-cycle-reinterpreted} simply reverses the outer loop over
simplices $\tsx$ and the inner loop over simplices $\ssx$. For each
$(p-1)$-simplex $\ssx$ the algorithm considers when its coefficient in cycle $w$ changes.
This happens exactly when we encounter the time $t_\tsx$ of one of its cofaces
$\tsx$ in cycle $z$. Accordingly, all such times $t_\tsx$ partition the interval
$[s,f]$ into intervals. The coefficient of $\ssx$ in $w$ differs between
adjacent intervals by the coefficient of $\tsx$ in $z$; see \cref{fig:cycle-coefficients}.
\cref{alg:lift-cycle-reinterpreted} goes through simplices $\ssx$ and times
$t_\tsx$ and adds cells $\ssx \times [l, t_\tsx]$ to the lifted cycle $\pz$ with
the same coefficients as \cref{alg:lift-cycle}.

\begin{figure}
    \centering
    \includegraphics{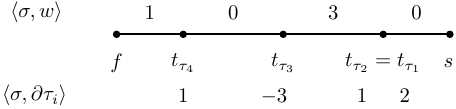}
    \caption{Coefficients of some $(p-1)$-dimensional simplex $\ssx$ in cycle
             $w$ in \cref{alg:lift-cycle} and in the
             boundary of its cofaces $\tsx_1, \tsx_2, \tsx_3, \tsx_4$. The coefficients
             in the cycle change exactly when \cref{alg:lift-cycle} encounters the
             chosen times $t_{\tsx_i}$ of the cofaces.
             \cref{alg:lift-cycle-reinterpreted} takes the $\ssx$-centric view
             and instead of keeping track of the full cycle $w$ computes the
             coefficients of $\ssx$ on the different intervals.
             We note that in this example $\ssx \notin \winit$,
             but $\ssx \in \wfinal$. In particular, we are lifting a relative cycle.}
    \label{fig:cycle-coefficients}
\end{figure}

\paragraph{Correctness.}
The correctness of \cref{alg:lift-cycle} is immediate and follows from the
following claim about the boundary structure,
We let $z_x^y$ be the restriction of cycle $z$ to the simplices $\tsx$ whose times
$t_\tsx$ lie in the interval $[x,y]$, i.e.,
$z_x^y = \sum_{\tsx \in z, t_\tsx \in [x,y]} \coeff{\tsx}{z} \cdot \tsx$.

\begin{claim}
    \label{clm:lifted-boundary}
    Given input cycle $z \in K$, the boundary of the lifted cycle $\pz$ satisfies
    $
        \bdry \pz = \winit \times s + \wfinal \times f,
    $
    where
    $
        \wfinal = \winit + \bdry z_b^d.
    $
\end{claim}
\begin{proof}
    The claim is immediate since $\wfinal$ is derived from $\winit$ by adding
    the boundaries of the simplices in the cycle whose time falls in the interval.
\end{proof}
It follows that if $\winit \subseteq K[s,s]$, then $\pz$ is a cycle in $\pK I$.
What is not immediate is why the chain $w \times [l,t_\tsx]$ in
Line \ref{line:z-update-2} of \cref{alg:lift-cycle} is in the prism $\pK$.

\begin{claim}
    Chain $w \times [l, t_\tsx]$ added in Line \ref{line:z-update-2} of
    \cref{alg:lift-cycle} is present in $\pK$.
\end{claim}
\begin{proof}
    After we process all the times $t_\tsx$ up to $l$, $w =
    \winit + \bdry z_s^l$.
    If $w \times [l, t_\tsx]$ is not in $\pK$, it means that there is some
    simplex $\ssx \in w$ such that $\ssx \times [l, t_\tsx] \notin \pK$.
    Assuming $\ssx$ is the first such simplex we encounter, either $\min (\ssx)$
    or $\max (\ssx)$ --- call it $x$ --- is in $[l,t_\tsx]$ (depending on whether $s > f$ or $s < f$).
    \cref{eq:face-nests} implies that $\ssx$ has no cofaces in
    the range $[x,f]$. Therefore, $\ssx \in \winit + \bdry z_s^x$ and
    $\ssx \notin \bdry z_x^f$. Therefore, either (1) $\ssx \in \bdry z$,
    or (2) $\ssx \in w$ and there is no $\tsx \in z$ with $\ssx \in \bdry \tsx$
    (this situation would happen right away at $l=s$).
    Since we assume $z$ is a relative cycle in $KI$, it cannot have any boundary
    inside the interval $I$ itself, so (1) is impossible.
    Since we assume $\bdry z \cap K[s,s] = \winit$, (2) is impossible.
\end{proof}

As stated, \cref{alg:lift-cycle} runs in quadratic time.

\end{document}